\documentclass[runningheads]{llncs}
\usepackage{graphicx}
\usepackage{amsmath}
\usepackage{subfigure}
\usepackage{longtable}
\usepackage{amssymb}
\usepackage{epsfig,latexsym,subfigure,cite}
\newcommand{\keywords}[1]{\par\addvspace\baselineskip
\noindent\keywordname\enspace\ignorespaces#1}

\providecommand{\xv}{\mathbf{x}} 
\providecommand{\uv}{\mathbf{u}} 
\providecommand{\sv}{\mathbf{s}} 
 \providecommand{\Xv}{\mathbf{X}}

\providecommand{\hv}{\mathbf{h}} \providecommand{\gv}{\mathbf{g}}
\providecommand{\rv}{\mathbf{r}} \providecommand{\bv}{\mathbf{b}}

\providecommand{\Wc}{{\mathcal W}} 
 
 \providecommand{\Cc}{{\mathcal C}}
\providecommand{\Pc}{{\mathcal P}} \providecommand{\Rc}{{\mathcal R}}
\providecommand{\Ac}{{\mathcal A}} \providecommand{\Sc}{{\mathcal S}}

\providecommand{\Yt}{\widetilde{Y}} 
 \providecommand{\zt}{\widetilde{z}}
\providecommand{\gvt}{\widetilde{\gv}} \providecommand{\hvt}{\widetilde{\hv}}

\begin{document}
\mainmatter              % start of the contributions

% first contribution

\title{Multiple Antenna Secure Broadcast over Wireless Networks}
\titlerunning{Multiple Antenna Secure Broadcast over Wireless Networks}

\author{Ruoheng Liu\thanks{This research was supported by the National Science Foundation under Grants ANI-03-38807 and CNS-06-25637.}
\and H. Vincent Poor}

% use the command \index{<name>} for index entries

\institute{Department of Electrical Engineering, Princeton University\\
Princeton, NJ 08544, USA,\\
\email{\{rliu,poor\}@princeton.edu}}

\maketitle              % typeset the title of the contribution

\begin{abstract}

In wireless data networks, communication is particularly susceptible to
eavesdropping due to its broadcast nature. Security and privacy systems have
become critical for wireless providers and enterprise networks. This paper
considers the problem of secret communication over the Gaussian broadcast
channel, where a multi-antenna transmitter sends independent confidential
messages to two users with \emph{perfect secrecy}. That is, each user would
like to obtain its own message reliably and confidentially. First, a computable
Sato-type outer bound on the secrecy capacity region is provided for a
multi-antenna broadcast channel with confidential messages. Next, a dirty-paper
secure coding scheme and its simplified version are described. For each case,
the corresponding achievable rate region is derived under the perfect secrecy
requirement. Finally, two numerical examples demonstrate that the Sato-type
outer bound is consistent with the boundary of the simplified dirty-paper
coding secrecy rate region.

\keywords{secure communication, broadcast channels, multiple antennas}

\end{abstract}

\section{Introduction}

The need for efficient, reliable, and secure data communication over wireless
networks has been rising rapidly for decades. Due to its broadcast nature,
wireless communication is particularly susceptible to eavesdropping. The
inherent problematic nature of wireless networks exposes not only the risks and
vulnerabilities that a malicious user can exploit and severely compromise the
network but also multiplies information confidentiality concerns with respect
to in-network terminals. Hence, security and privacy systems have become
critical for wireless providers and enterprise networks.

In this work, we consider multiple antenna secure broadcast in wireless
networks. This research is inspired by the seminal paper \cite{Wyner:BSTJ:75},
in which Wyner introduced the so-called {\it wiretap channel} and proposed an
information theoretic approach to secure communication schemes. Under the
assumption that the channel to the eavesdropper is a degraded version of that
to the desired receiver, Wyner characterized the capacity-secrecy tradeoff for
the discrete memoryless wiretap channel and showed that secure communication is
possible without sharing a secret key. Later, the result was extended by
Csisz{\'{a}r and K{\"{o}rner who determined the secrecy capacity for the
non-degraded broadcast channel (BC) with a single confidential message set
intended for one of the users \cite{Csiszar:IT:78}.

%Latter, the result was extended to the scalar Gaussian wiretap channel in
%\cite{Cheong:IT:78}.

In more general wireless network scenarios, secure communication may involve
multiple users and multiple antennas. Motivated by wireless communications,
where transmitted signals are broadcast and can be received by all users within
the communication range, a significant research effort has been invested in the
study of the information-theoretic limits of secure communication in different
wireless network environments
\cite{Oohama:ITW:01,Csiszar:IT:04,Tekin:ISIT:06,Liang06it,Liu:ISIT:06,Liu:Allerton:06,Lai:IT:06,Li:CISS:07}.

This issues motivate us to study the multi-antenna Gaussian broadcast channel
with confidential messages (MGBC-CM), in which independent confidential
messages from a multi-antenna transmitter are to be communicated to two users.
The corresponding broadcast communication model is shown in Fig.~\ref{fig:gbc}.
Each user would like to obtain its own message reliably and confidentially.
\begin{figure}[t]
 \centerline{\includegraphics[width=0.65\linewidth,draft=false]{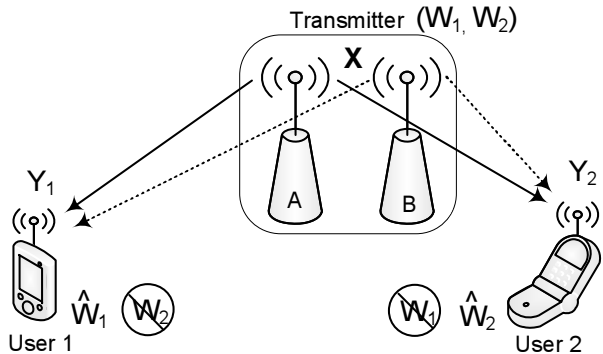}} \caption{
Multiple-antenna Gaussian broadcast channel with confidential message}
  \label{fig:gbc}
\end{figure}

To give insight into this problem, we first consider a single antenna Gaussian
broadcast channel (GBC). Note that this channel is degraded~\cite{Cover}, which
means that if a message can be successfully decoded by the inferior user, then
the superior user is also ensured of decoding it. Hence, the secrecy rate of
the inferior user is zero and this problem is reduced to the scalar Gaussian
wiretap channel problem~\cite{Cheong:IT:78} whose secrecy capacity is now the
maximum rate achievable by the superior user. This analysis gives rise to the
question: can the transmitter, in fact, communicate with both users
confidentially at nonzero rate under some other conditions? Roughly speaking,
the answer is in the affirmative. In particular, the transmitter can
communicate when equipped with sufficiently separated multiple antennas.

We here have two goals motivated directly by questions arising in practice. The
first is to determine the condition under which both users can obtain their own
messages reliably and confidentially. This is equivalent to evaluating the
secrecy capacity region for the MGBC-CM. The second is to show \emph{how} the
transmitter should broadcast securely, which is equivalent to designing an
achievable secure coding scheme. To this end, a computable Sato-type outer
bound on the secrecy capacity region is developed for the MGBC-CM in
Sec.~\ref{sec:out}. Next, a dirty-paper secure coding scheme and its simplified
version are described. For each case, the corresponding achievable rate region
is derived under perfect secrecy requirement in Sec.~\ref{sec:in}. Finally, two
numerical examples demonstrate that the Sato-type outer bound is consistent
with the boundary of the simplified dirty-paper coding (DPC) secrecy rate
region in Sec.~\ref{sec:ex}.

\section{System Model}

We consider the communication of confidential messages to two users over a
Gaussian broadcast channel via $t$ transmit-antennas. Each user is equipped
with a single receive-antenna. The received signals at user~1 and user~2 are
are modeled as
\begin{align}
y_{1,i}&=\hv^{H} \xv_i+z_{1,i} \notag\\
y_{2,i}&=\gv^{H} \xv_i+z_{2,i},\qquad i=1,\ldots,n \label{eq:miso}
\end{align}
where $\xv_i\in \mathbb{C}^{t}$ is the transmitted vector at time $i$,
$z_{1,i}$ and $z_{2,i}$ correspond to two independent, zero-mean,
unit-variance, complex Gaussian noise sequences, and $\hv,\,\gv\in
\mathbb{C}^{t}$ are channel attenuation vectors corresponding to user~1 and
user~2, respectively. The channel input is constrained by
\begin{align}
\frac{1}{n}\sum_{i=1}^{n} |\xv_i|^2\le A
\end{align}
where $A$ is the total transmit energy per channel use. We also assume that
both the transmitter and receivers are aware of the attenuation vectors $\hv$
and $\gv$.

As shown in Fig.~\ref{fig:gbc}, the transmitter intends to send an independent
confidential message $W_k\in\{1,\dots,M_k\}\triangleq \Wc_k$ to the respective
user $k\in\{1,\,2\}$ in $n$ channel uses. To increase the randomness of
transmitted messages, we consider a stochastic encoder at the transmitter. More
explicitly, the encoder is specified by a matrix of conditional probability
density $f(\xv^{n}|w_1,w_2)$, where $\xv^{n}=[\xv_1,\ldots,\xv_n]$ and $w_k\in
\Wc_k$. In other words, $f(\xv^{n}|w_1,w_2)$ is the probability density
associated with the conditional probability that the messages $(w_1,w_2)$ are
encoded as the channel input $\xv^{n}$.

The decoding function at user $k$ is a mapping $\phi_k\, :\,
\mathbb{C}^{n}\rightarrow \Wc_k$. The secrecy levels with respect to the
confidential messages $W_1$ and $W_2$ are measured, respectively, at receivers
$1$ and $2$ with respect to the normalized equivocations
\begin{equation} \label{eq:nq}
\frac{1}{n}H(W_2|Y_1^{n},W_1)\quad \text{and} \quad
\frac{1}{n}H(W_1|Y_2^{n},W_2).
\end{equation}

An $(M_1,M_2,n,P_e)$ code for the broadcast channel consists of the encoding
function $f$, decoding functions $\phi_1$, $\phi_2$, and the maximum error
probability $P_e\triangleq \max \{P_{e,1},\,P_{e,2}\}$, where $P_{e,k}$ is the
error probability for user $k$ given by
\begin{align}
P_{e,k}&=P\bigl[\phi_k(Y_k^{n})\neq w_k\bigr].
\end{align}

A rate pair $(R_1,\, R_2)$ is said to be achievable for the broadcast channel
with confidential messages if, for any $\epsilon>0$, there exists an $(M_1,
M_2, n, P_e)$ code that satisfies $P_e\le \epsilon$, $M_k\ge 2^{nR_k}$, for
$k=1,2$, and the perfect secrecy requirement
\begin{align}
H(W_1)-H(W_1|Y_2^{n},W_2)&\le n\epsilon \quad \text{and} \quad
H(W_2)-H(W_2|Y_1^{n},W_1)\le n\epsilon \label{eq:equiv}.
\end{align}

The secrecy capacity region of $\Cc_{\rm BCC}$ of the MGBC-CM is the closure of
the set of all achievable rate pairs $(R_1,\,R_2)$.

\section{Outer Bound on the Secrecy Capacity Region} \label{sec:out}

\subsection{Sato-Type Outer Bound}

We first consider a Sato-type bound that can be applied to both discrete
memoryless and Gaussian broadcast channels with confidential messages (BC-CM).
Let $\Pc$ be the set of channels $P_{\Yt_1,\Yt_2|\Xv}$ that have the same
marginal distributions as $P_{Y_1,Y_2|\Xv}$, i.e.,
\begin{align}
P_{\Yt_1|\Xv}(y_1|\xv)=P(y_1|\xv) \quad \text{and} \quad
P_{\Yt_2|\Xv}(y_2|\xv)=P(y_2|\xv)
\end{align}
for all $y_1$, $y_2$ and $\xv$. Let $\Rc_{\rm
O}\bigl(P_{\Yt_1,\Yt_2|\Xv},\,P_{\Xv}\bigr)$ denote the union of all rate pairs
$(R_1,\,R_2)$ satisfying
\begin{align}
R_1&\le I(\Xv;\Yt_1|\Yt_2) \quad \text{and} \quad R_2\le I(\Xv;\Yt_2|\Yt_1)
\end{align}
for given distributions $P_{\Xv}$ and $P_{\Yt_1,\Yt_2|\Xv}$.

\begin{theorem} \label{thm:out1} %{\bf (outer bound)}
The secrecy capacity region $\Cc_{\rm BCC}$ of the BC-CM satisfies
\begin{align}
\Cc_{\rm BCC} \subseteq \bigcap_{P_{\Yt_1,\Yt_2|\Xv}\in \Pc}
\left\{\bigcup_{P_{\Xv}} \Rc_{\rm
O}\bigl(P_{\Yt_1,\Yt_2|\Xv},\,P_{\Xv}\bigr)\right\}. \label{eq:Sato}
\end{align}
\end{theorem}
\begin{proof}
See the Appendix.
\end{proof}
\begin{remark}
The outer bound (\ref{eq:Sato}) follows by letting the users decode the message
in a \emph{cooperative} manner, while evaluating the secrecy level in an
individual manner. Hence, the eavesdropped signal is always a degraded version
of the entire received signal. This permits the use of the wiretap channel
result of \cite{Wyner:BSTJ:75}.
\end{remark}
\begin{remark}
Although Theorem~\ref{thm:out1} is based on a \emph{degraded} argument, the
outer bound (\ref{eq:Sato}) can be applied to \emph{general} broadcast channels
with confidential messages.
\end{remark}

\subsection{Sato-Type Outer Bound for the Gaussian BC-CM}

For the Gaussian channel case, we can simplify the outer bound (\ref{eq:Sato})
using the following steps. First, the channel family  $\Pc$ is the set of
channels~(\ref{eq:miso}) where $z_1$ and $z_2$  are replaced by arbitrarily
correlated, zero-mean, unit-variance, Gaussian random variables $\zt_1$ and
$\zt_2$ with covariance $\nu$. Furthermore, we consider a new coordinate
transform as the setting of \cite{Li:CISS:07} and rewrite the broadcast channel
model (\ref{eq:miso}) as follows:
\begin{align}
y_1&=|\hv| s_1+\zt_1\notag\\
y_2&=\alpha |\gv| s_1+\sqrt{1-|\alpha|^2} |\gv| s_2+\zt_2. \label{eq:model}
\end{align}
where
\begin{align} \alpha&=\frac{\gv^{H}\hv}{|\hv||\gv|},\quad
s_1=\frac{\hv^{H}\xv}{|\hv|}\quad \text{and}~ s_2=\frac{(|\hv|\gv-\alpha
|\gv|\hv)^{H}\xv}{\sqrt{1-|\alpha|^2} |\hv||\gv|}.
\end{align}
From now on, we define
\begin{align}
\sv=\left[\begin{matrix}{s_1} \\{s_2}\end{matrix}\right],\quad
\hvt=|\hv|\bigl[1,\,0\bigr]^{H}\quad \text{and} \quad
\gvt=|\gv|\bigl[\alpha,\,\sqrt{1-|\alpha|^2}\bigr]^{H}.
\end{align}
The vector $\sv$ can be interpreted as the projection of $\xv$ onto the
subspace spanned by $\hv$ and $\gv$, or more precisely, the projection of $\xv$
onto the $2$ orthonormal bases
\begin{align}
\rv_1=\frac{\hv}{|\hv|} \quad \text{and} \quad \rv_2=\frac{|\hv|\gv-\alpha
|\gv|\hv}{\sqrt{1-|\alpha|^2} |\hv||\gv|}.
\end{align}
Since the projection operation cannot increase the length of a vector, the
covariance matrix $K_{S_1,S_2}$ satisfies the input constraint ${\rm
tr}(K_{S_1,S_2})\le A$. We also note that the Markov chain property $\Xv
\rightarrow (S_1,S_2) \rightarrow (\Yt_1,\Yt_2)$ holds, and hence,
\begin{align}
I(\Xv;\Yt_1|\Yt_2)=I(S_1,S_2;\Yt_1|\Yt_2)\quad\text{and}\quad
I(\Xv;\Yt_2|\Yt_1)=I(S_1,S_2;\Yt_2|\Yt_1).
\end{align}
Following \cite{Cheong:IT:78}, it can be shown that Gaussian input
distributions maximize ${\Rc}_{\rm O}$ by applying the maximum-entropy theorem
\cite{Cover}. Hence, we restrict attention to a zero-mean Gaussian pair
$(S_1,\,S_2)$ with the covariance matrix $K_{S_1,S_2}$. These facts are
summarized in the following.

For a multi-antenna Gaussian broadcast channel, the rate region ${\Rc}_{\rm O}$
is a function of the noise covariance $\nu$ and the input covariance matrix
$K_{S_1,S_2}$, i.e., ${\Rc}_{\rm O}(\nu,\,K_{S_1,\,S_2})$ is the union of all
rate pairs $(R_1,\,R_2)$ satisfying
\begin{align}
R_1&\le \log_2\frac{K_{\Yt_1,\Yt_2}} {(1-|\nu|^2)(1+\gvt^{H}K_{S_1,S_2}\gvt)}
\end{align}
and
\begin{align}
R_2&\le \log_2\frac{K_{\Yt_1,\Yt_2}} {(1-|\nu|^2)(1+\hvt^{H}K_{S_1,S_2}\hvt)}
\end{align}
where $K_{\Yt_1,\Yt_2}$ is the covariance matrix of $\Yt_1$ and $\Yt_2$. This
covariance matrix is given explicitly by
\begin{align}
K_{\Yt_1,\Yt_2}&=\sqrt{1-|\alpha|^2}|\hvt|^2|\gvt|^2|K_{S_1,S_2}|+\hvt^{H}K_{S_1,S_2}\hvt+\gvt^{H}K_{S_1,S_2}\gvt\notag\\
&\qquad -\gvt^{H}K_{S_1,S_2}(\nu\hvt) - (\nu\hvt)^{H}K_{S_1,S_2}\gvt+1-|\nu|^2.
\end{align}
\begin{theorem} \label{thm:outG}
For an MGBC-CM, the secrecy capacity region $\Cc_{\rm BCC}$ satisfies
\begin{align}
\Cc_{\rm BCC} \subseteq \bigcap_{0\le|\nu|\le1} \left\{\bigcup_{K_{S_1,S_2}\in
\Ac} {\Rc}_{\rm O}(\nu,\,K_{S_1,S_2})\right\} \label{eq:Sato2}
\end{align}
where $\Ac$ is the set of all covariance matrices satisfying the input
constraint ${\rm tr}(K_{S_1,S_2})\le A$.
\end{theorem}
\begin{proof}
Theorem~\ref{thm:outG} follows from Theorem~\ref{thm:out1} and the fact that
the optimum input distribute is Gaussian.
\end{proof}
\begin{remark}
Theorem~\ref{thm:outG} describes a computable outer bound of the secrecy
capacity region for the MGBC-CM. The bound (\ref{eq:Sato2}) provides a
benchmark for evaluating the goodness of achievable coding schemes described in
next section.
\end{remark}

\section{Inner Bound and Achievable Coding Scheme} \label{sec:in}

\subsection{Inner Bound for the BC-CM}
An inner bound for the BC-CM has been established in \cite[Theorem~3]{Liu07it}.
Here we first review the result as follows. Let $V_1$ and $V_2$ be auxiliary
random variables. We define $\Omega$ as the class of joint probability
densities $p(v_1,v_2,\xv,y_1,y_2)$ that factor as
$p(v_1,v_2)p(\xv|v_1,v_2)p(y_1,y_2|\xv).$ Let ${\Rc}_{\rm I}(\pi)$ denote the
union of all $(R_1,R_2)$ satisfying
\begin{align}
0 &\le R_1 \le I(V_1;Y_1)-I(V_1;Y_2|V_2)-I(V_1;V_2) \label{eq:BC-IN-R1}
\end{align} and
\begin{align}
 0 &\le R_2 \le
I(V_2;Y_2)-I(V_2;Y_1|V_1)-I(V_1;V_2) \label{eq:BC-IN-R2}
\end{align}
for a given joint probability density $\pi\in \Omega$.

\begin{theorem} {\rm \cite[Theorem~3]{Liu07it}} \label{thm:InBC} %{\bf (inner bound)}
Any rate pair
\begin{align}
(R_1,R_2)\in {\rm co} \left\{\bigcup_{\pi\in\Omega} \Rc_{\rm I}(\pi)\right\}
\label{eq:inner}
\end{align}
is achievable for the broadcast channel with confidential messages, where ${\rm
co}\{\Sc\}$ denotes the convex hull of the set $\Sc$.
\end{theorem}

The proof of Theorem~\ref{thm:InBC} can be found in \cite{Liu07it}. Here, we
provide an alternative view on this result. The best known achievable rate for
a general BC was found by Marton in \cite{marton:IT:77}. For a given joint
density $p(v_1,v_2,\xv)$, the Marton sum rate (without a common rate) is given
by
$$I(V_1;Y_1)+I(V_2;Y_2)-I(V_1;V_2).$$ On the other hand, the total (both
intended and eavesdropped) information rate obtained by user $2$ is bounded by
$I(V_1,V_2;Y_2).$ This implies that to satisfy the perfect secrecy requirement,
the achievable secrecy rate of user $1$ can be written as
\begin{align*}
R_1\le [I(V_1;Y_1)+I(V_2;Y_2)-I(V_1;V_2)]-I(V_1,V_2;Y_2)
\end{align*}
which leads to the bound in (\ref{eq:BC-IN-R1}).

\begin{remark}
We note that for a broadcast channel, we can employ joint encoding at the
transmitter. However, to preserve confidentiality, both achievable rates in
(\ref{eq:BC-IN-R1}) and (\ref{eq:BC-IN-R2}) include a penalty term
$I(V_1;V_2)$. Hence, compared with Marton's sum rate bound \cite{marton:IT:77},
here, we pay ``double'' to jointly encode at the transmitter.
\end{remark}

\subsection{Dirty Paper Coding Scheme for the MGBC-CM}
The achievable strategy in Theorem~\ref{thm:InBC} introduces a {\it
double-binning} coding scheme that enables both joint encoding at the
transmitter by using Slepian-Wolf binning \cite{Slepian:IT:73} and preserving
confidentiality by using random binning. However, when the rate region
(\ref{eq:inner}) is used as a constructive technique, it not clear how to
choose the auxiliary random variables $V_1$ and $V_2$ to implement the double
binning codebook, and hence, one has to ``guess'' the density of
$p(v_1,v_2,\xv)$. Here, we employ the DPC strategy to study the achievable
secrecy rate region.

For the MGBC-CM, we focus on the new coordinate channel model (\ref{eq:model})
and let
\begin{align}
\xv=s_1\rv_1+s_2 \rv_2.
\end{align}
Hence, the vector $\sv$ can be viewed as a precoded signal for $\xv$. We
generate signal $\sv$ by dirty paper encoding with Gaussian codebooks
\cite{GelfandPinsker80,Costa:IT:83} as follows.

First, sperate the precoded signal $\sv$ into two vectors so that
\begin{align}
\uv_1=\left[\begin{matrix}{u_{1,1}} \\{u_{1,2}}\end{matrix}\right],\quad
\uv_2=\left[\begin{matrix}{u_{2,1}}
\\{u_{2,2}}\end{matrix}\right] \quad \text{and} \quad \uv_1+\uv_2=\sv.
\end{align}
Let $U_1$ and $U_2$ denote random variables corresponding to $\uv_1$ and
$\uv_2$, respectively. We choose $U_1$ and $U_2$ as well as auxiliary random
variables $V_1$ and $V_2$ as follows:
\begin{align}
U_1&\thicksim \mathcal{CN}(0,K_{U_1}),\notag\\
U_2&\thicksim \mathcal{CN}(0,K_{U_2}),~\text{independent of }U_1\notag\\
V_1&=U_1\quad \text{and}\quad V_2=U_2+\bv \gvt^{H}U_1 \label{eq:rvs}
\end{align}
where $K_{U_1}$ and $K_{U_2}$ are covariance matrices of $U_1$ and $U_2$,
respectively, and
\begin{align}
\bv=K_{U_2}\gvt(1+\gvt^{H} K_{U_2} \gvt)^{-1}.
\end{align}
Based on the conditions (\ref{eq:rvs}) and Theorem~\ref{thm:InBC}, we obtain a
DPC rate region for the MGBC-CM as follows.

\begin{theorem}\label{thm:GBCin} {\rm [DPC region]}
Let  ${\Rc}_{\rm I}^{\rm [DPC]}(K_{U_1},K_{U_2})$ denote the union of all
$(R_1,R_2)$ satisfying
\begin{align}
0&\le R_1\le \log_2 \frac{1+\hvt^{H} (K_{U_1}+K_{U_2}) \hvt}{1+\hvt^{H} K_{U_2}
\hvt}-\log_2 \frac{1+\gvt^{H} (K_{U_1}+K_{U_2}) \gvt}{1+\gvt^{H} K_{U_2}\gvt}
\end{align}
and
\begin{align}
0&\le R_2\le \log_2 \frac{1+\gvt^{H} K_{U_2} \gvt}{1+\hvt^{H} K_{U_2} \hvt}.
\end{align}
Then, any rate pair
\begin{align}
(R_1,R_2)\in {\rm co} \left\{\bigcup_{{\rm tr}(K_{U_1}+K_{U_2})\le A} \Rc_{\rm
I}(K_{U_1},K_{U_2})\right\} \label{eq:dpc-r}
\end{align}
is achievable for the multi-antenna Gaussian broadcast channel with
confidential messages.
\end{theorem}
\begin{proof}
See the Appendix.
\end{proof}
\begin{remark}
In general, the set of achievable secrecy rates may be increased by considering
another new coordinate channel with the bases
\begin{align}
\rv'_1=\frac{\gv}{|\gv|} \quad \text{and} \quad
\rv'_2=\frac{|\gv|\hv-\alpha^{*} |\hv|\gv}{\sqrt{1-|\alpha|^2} |\hv||\gv|}
\end{align}
and reversing the roles of user $1$ and user $2$.
\end{remark}

\subsection{Simplified Dirty Paper Coding Scheme}

The DPC secrecy rate region (\ref{eq:dpc-r}) requires optimization of the
covariance matrices $K_{U_1}$ and $K_{U_2}$. Here, we consider a simplified DPC
scheme as follows. Let
\begin{align}
A_{i,j}=E[|U_{i,j}|^2], \quad \text{for}~i=1,2~\text{and}~j=1,2 \label{eq:aa}
\end{align}
where $U_{i,j}$ denotes the random variable corresponding to $u_{i,j}$. The
channel input power constraint implies that
\begin{align}
A_{1,1}+A_{1,2}+A_{2,1}+A_{2,2}\le A.
\end{align}
In particular, we choose the normalized correlation coefficients as
\begin{align}
\rho_1\triangleq \frac{E[U_{1,1}U_{1,2}^{*}]}{\sqrt{A_{1,1}A_{1,2}}}
=-\frac{\alpha^{*}}{|\alpha|} \quad\text{and}\quad \rho_2\triangleq
\frac{E[U_{2,1}U_{2,2}^{*}]}{\sqrt{A_{2,1}A_{2,2}}}
=\frac{\alpha^{*}}{|\alpha|}. \label{eq:rho}
\end{align}

Now, we describe a simplified DPC secrecy rate region based on the setting
(\ref{eq:aa})-(\ref{eq:rho}). Let
\begin{align}
&f_1(A_{1,1},A_{1,2},A_{2,1},A_{2,2})=\notag\\
&\quad\frac{[1+|\hvt|^2(A_{1,1}+A_{2,1})]f_2(A_{2,1},A_{2,2})}
{1+|\gvt|^2\left[\sqrt{ \beta
A_{1,1}}-\sqrt{(1-\beta)A_{1,2}}\right]^2+|\gvt|^2\left[\sqrt{ \beta
A_{2,1}}+\sqrt{(1-\beta)A_{2,2}}\right]^2}
\end{align}
and
\begin{align}
f_2(A_{2,1},A_{2,2})=\frac{1+|\gvt|^2\left[\sqrt{ \beta
A_{2,1}}+\sqrt{(1-\beta) A_{2,2}}\right]^2}{1+|\hvt| A_{2,1}}
\end{align}
where $\beta=|\alpha|^2$.
\begin{lemma}\label{lem:SDPC} {\rm [simplified DPC region]}
Let  ${\Rc}_{\rm I}^{\rm [S-DPC]}(A_{1,1},A_{1,2},A_{2,1},A_{2,2})$ denote the
union of all $(R_1,R_2)$ satisfying
\begin{align}
0&\le R_1\le \log_2 f_1(A_{1,1},A_{1,2},A_{2,1},A_{2,2})
\end{align}
and
\begin{align}
0&\le R_2\le \log_2 f_2(A_{2,1},A_{2,2}).
\end{align}
Then, any rate pair
\begin{align}
(R_1,R_2)\in {\rm co} \left\{\bigcup_{A_{1,1}+A_{1,2}+A_{2,1}+A_{2,2}\le A}
\Rc_{\rm I}(A_{1,1},A_{1,2},A_{2,1},A_{2,2})\right\} \label{eq:s-dpc-r}
\end{align}
is achievable for the multi-antenna Gaussian broadcast channel with
confidential messages.
\end{lemma}
\begin{remark}
To calculate the simplified DPC secrecy rate region, we need only to allocate
the total power into the precoded signals $u_{1,1}$, $u_{1,2}$, $u_{2,1}$ and
$u_{2,2}$. This significantly reduces the computational complexity.
\end{remark}

\subsection{A Special Case}

A special case of the MGBC-CM model is the Gaussian MISO wiretap channel, where
the transmitter sends confidential information to only one user (e.g., user~2)
and treats another user (e.g., user 1) as an eavesdropper. In this case, we set
$K_{U_1}=0$, so that $K_{S_1,S_2}=K_{U_2}$. Now, applying
Theorems~\ref{thm:outG} and \ref{thm:GBCin}, we obtain the following upper and
lower bounds on the secrecy capacity of the Gaussian MISO wiretap channel:
\begin{align}
\Cc_{\rm MISO}&\le\min_{0\le|\nu|\le1}\max_{{\rm tr}(K_{S_1,S_2})\le A}
\log_2\frac{K_{\Yt_1,\Yt_2}} {(1-|\nu|^2)(1+\hvt^{H}K_{S_1,S_2}\hvt)}
\end{align}
and
\begin{align}
\Cc_{\rm MISO}&\ge\max_{{\rm tr}(K_{S_1,S_2})\le A}\log_2 (1+\gvt^{H}
K_{S_1,S_2} \gvt)(1+\hvt^{H} K_{S_1,S_2} \hvt)^{-1}.
\end{align}
It can been verified that the two bounds meet to describe the secrecy capacity
of the MISO wiretap channel, which is consistent with the result in
\cite{Wornell:ISIT:07}. In other words, the corner points of the Sato-type
outer bound (\ref{eq:Sato2}) and the DPC achievable rate region
(\ref{eq:dpc-r}) are identical.

\section{Numerical Examples} \label{sec:ex}

In this section, we study two numerical examples to illustrates the secrecy
rate region of the MGBC-CM. For simplicity, we assume that the GBC has real
input and output alphabets and $\alpha$ is real too. Under this setting, all
calculated secrecy rate values are divided by $2$.

In the first example, we consider the following GBC with a large $\alpha$,
\begin{align}
y_1&=s_1+\zt_1\notag\\
y_2&=2(0.9s_1+\sqrt{1-0.81}s_2)+\zt_2 \label{eq:ex1}
\end{align}
i.e., $|\hvt|=1$, $|\gvt|=2$ and $\alpha=0.9$. The total power constraint is
set to $A=10$.
\begin{figure}[t]
 \centerline{\includegraphics[width=0.8\linewidth,draft=false]{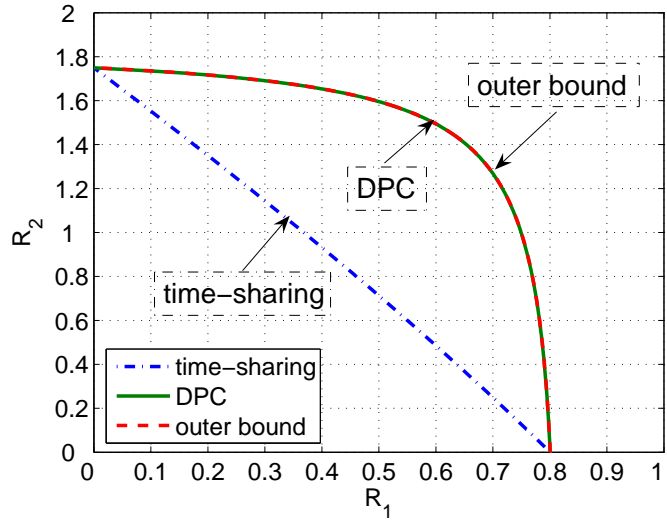}}
\caption{Comparison of the Sato-type outer bound and secrecy rate regions
achieved by time-sharing and simplified DPC schemes for the example MGBC-CM in
(\ref{eq:ex1})}
  \label{fig:sim1}
\end{figure}
Fig.~\ref{fig:sim1} depicts inner and outer bounds on the secrecy capacity for
the example MGBC-CM described in (\ref{eq:ex1}). We compare the Sato-type outer
bound (indicated by the dashed line) with the secrecy rate regions achieved by
the simplified DPC coding scheme (indicated by the solid line) and the
time-sharing scheme (indicated by the dash-dot line). Surprisingly, we observe
that not only the corner points but also the boundary of the simplified DPC
secrecy rate region (\ref{eq:s-dpc-r}) is identical with the Sato-type outer
bound (\ref{eq:Sato2}). Furthermore, Fig.~\ref{fig:sim1} demonstrates that the
time-sharing scheme is strictly suboptimal for providing the secrecy capacity
region.

\begin{figure}[t]
 \centerline{\includegraphics[width=0.8\linewidth,draft=false]{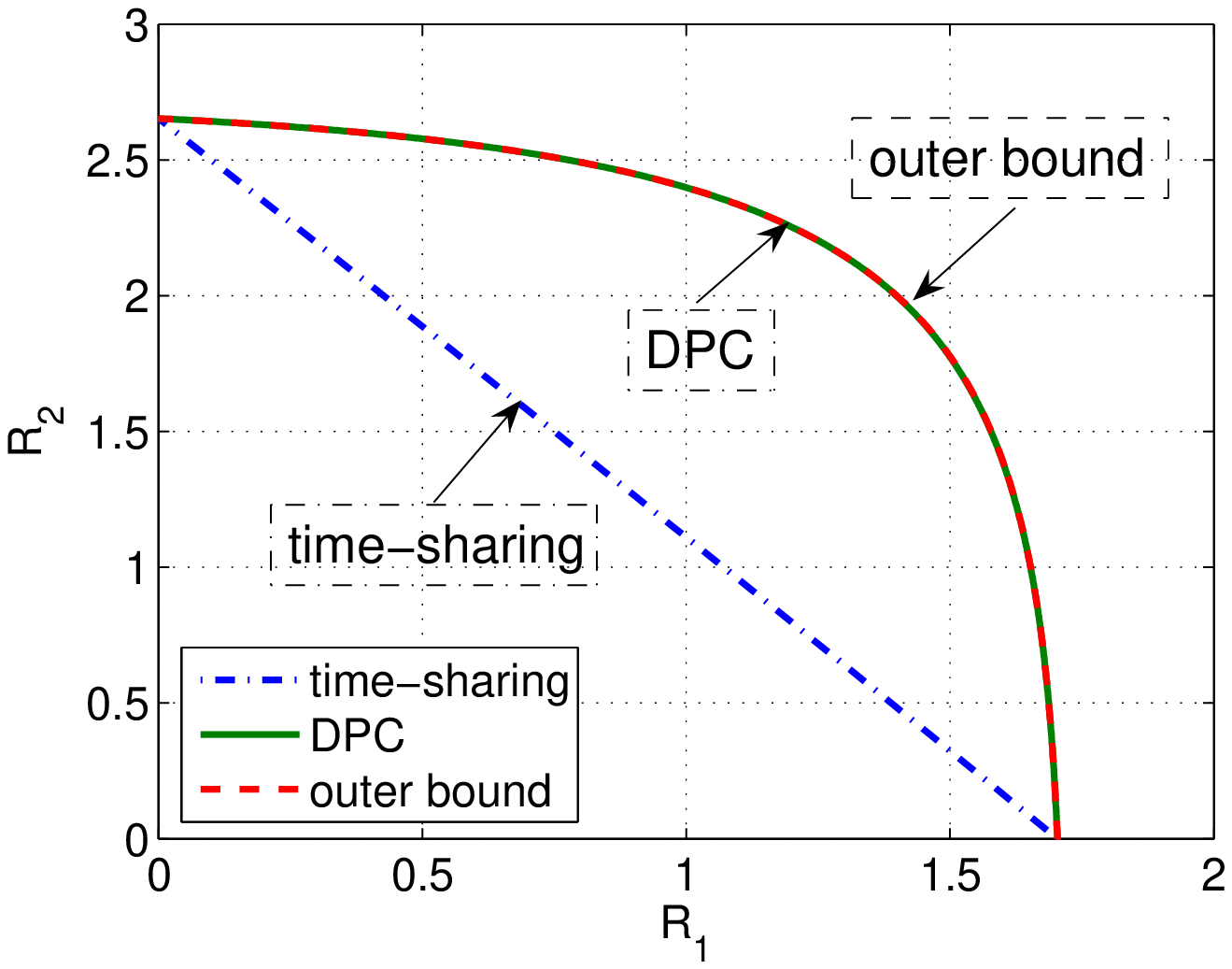}}
\caption{Comparison of the Sato-type outer bound and secrecy rate regions
achieved by time-sharing and simplified DPC schemes for the example MGBC-CM in
(\ref{eq:ex2})}
  \label{fig:sim2}
\end{figure}
In the second example, we consider a broadcast channel with a small $\alpha$
\begin{align}
y_1&=s_1+\zt_1\notag\\
y_2&=2(0.2s_1+\sqrt{1-0.04}s_2)+\zt_2 \label{eq:ex2}
\end{align}
i.e., $|\hvt|=1$, $|\gvt|=2$ and $\alpha=0.2$. The total power $A=10$.
Fig.~\ref{fig:sim2} illustrates that, again, the boundary of the simplified DPC
secrecy rate region (\ref{eq:s-dpc-r}) is consistent with the Sato-type outer
bound (\ref{eq:Sato2}).

\section{Conclusion}
In this paper, we have investigated outer and inner bounds on the secrecy
capacity region of a generally non-degraded Gaussian broadcast channel with
confidential messages for two users, where the transmitter has $t$ antennas and
each user has one antenna. For this model, we have introduced a computable
Sato-type outer bound and proposed a dirty-paper secure coding scheme. Using
numerical examples, we have illustrated that the boundary of the simplified DPC
secrecy rate region is consistent with the Sato-type outer bound.

Based on this observation, we conjecture that the dirty-paper secure coding
strategy achieves the secrecy capacity region of the MGBC-CM.

\section*{Appendix}

\begin{proof}  {\bf(Theorem~\ref{thm:out1})}
Here we prove Theorem~\ref{thm:out1} and derive the outer bound for $R_1$. The
outer bound for $R_2$ follows by symmetry.

The secrecy requirement (\ref{eq:equiv}) implies that
\begin{align}
nR_1= H(W_1) &\le H(W_1|Y_2^{n},W_2)+n\epsilon  \label{eq:r1}.
\end{align}
On the other hand, Fano's inequality and $P_e\le \epsilon$ imply that
\begin{align}
H(W_1|Y_1^{n}) &\le \epsilon \log(M_1-1)+h(\epsilon) \triangleq n\delta_1.
\label{eq:sd1}
\end{align}
where $h(x)$ is the binary entropy function. Now, we can bound (\ref{eq:r1}) as
follows
\begin{align}
nR_1& \le H(W_1|Y_2^{n})+n\epsilon \label{eq:rele} \\
&\le H(W_1|Y_2^{n})-H(W_1|Y_1^{n})+n(\delta_1+\epsilon) \label{eq:cr01} \\
&\le H(W_1|Y_2^{n})-H(W_1|Y_1^{n},Y_2^{n})+n(\delta_1+\epsilon) \\
&= I(W_1;Y_1^{n}|Y_2^{n})+n(\delta_1+\epsilon)\\
&\le I(\Xv^{n};Y_1^{n}|Y_2^{n})+n(\delta_1+\epsilon)\\
%&=\sum_{i=1}^{n} I(\Xv^{n};Y_{1,i}|Y_2^{n},Y_2^{i-1})+n(\delta_1+\epsilon) \\
&\le \sum_{i=1}^{n} I(\Xv_i;Y_{1,i}|Y_{2,i})+n(\delta_1+\epsilon).
\end{align}

Note that the decoding error probability and the equivocation rate at each user
depend only on the marginal probability densities $P(y_1|\xv)$ and
$P(y_2|\xv)$. Hence, one can replace $Y_1$ and $Y_2$ by $\Yt_1$ and $\Yt_2$.
Therefore, we have the desired result.
\end{proof}

\begin{proof}  {\bf(Theorem~\ref{thm:GBCin})}
We first check the power constraint. Since $U_1$ and $U_2$ are independent and
$$\sv=[s_1,s_2]^{\top}=\uv_1+\uv_2,$$ the covariance matrices $K_{U_1}$ and $K_{U_2}$
satisfy
\begin{align} {\rm tr}(K_{U_1}+K_{U_2})={\rm tr}(K_{S_1,S_2})\le A.
\end{align}

Following from \cite[Theorem~1]{Yu:IT:04} and using the setting in
(\ref{eq:rvs}), we can immediately obtain the well-known \emph{successive
encoding} result:
\begin{align}
I(V_1;Y_1)&=\log_2 \frac{1+\hvt^{H} (K_{U_1}+K_{U_2}) \hvt}{1+\hvt^{H} K_{U_2}
\hvt}\label{eq:pf1}
\end{align}
and
\begin{align}
I(V_2;Y_2)-I(V_1;V_2)&= \log_2 (1+\gvt^{H} K_{U_2} \gvt).
\end{align}
Since $V_1=U_1$ is independent of $U_2$ and $V_2=U_2+\bv\gvt^{H} U_1$, we have
\begin{align}
I(V_2;Y_1|V_1)&=I(U_2+\bv\gvt^{H} U_1;\,Y_1|U_1)\notag\\
&=I(U_2;Y_1|U_1) \\
&=\log_2(1+\hvt^{H} K_{U_2} \hvt).
\end{align}
Moreover, we note that
\begin{align}
I(V_1;Y_2|V_2)+I(V_1;V_2)&=I(V_1,V_2;Y_2)-[I(V_2;Y_2)-I(V_1;V_2)]\notag\\
&=\log_2\frac{1+\gvt^{H} (K_{U_1}+K_{U_2}) \gvt}{1+\gvt^{H}
K_{U_2}\gvt}.\label{eq:pf2}
\end{align}
Substituting (\ref{eq:pf1})-(\ref{eq:pf2}) into (\ref{eq:BC-IN-R1}) and
(\ref{eq:BC-IN-R2}), we have the desired result.
\end{proof}

\bibliographystyle{IEEEtran}
\bibliography{secrecy}

\end{document}